\documentclass[11pt,a4paper]{article}

\usepackage[margin=2cm]{geometry}
\usepackage[utf8]{inputenc}
\usepackage{amsmath}
\usepackage{amssymb}
\usepackage{amsthm}
\usepackage{dsfont}
\usepackage{color}
\usepackage{amsmath}
\usepackage{amsfonts}
\usepackage{amssymb}
\usepackage{bbm}
\usepackage{relsize}
\usepackage{enumitem}
\usepackage[procnumbered, linesnumbered,algoruled]{algorithm2e}
\usepackage[compact]{titlesec}
\usepackage{graphicx}
\usepackage[hidelinks,colorlinks,allcolors=blue]{hyperref}
\graphicspath{ {./images/} }

\newcommand{\ceil}[1]{{\left\lceil#1  \right\rceil}}
\newcommand{\comment}[1]{}

\newcommand{\rank}{\mathsf{rank}}

\newcommand{\cA}{{\mathcal{A}}}

\newcommand{\cI}{{\mathcal{I}}}
\newcommand{\cm}{{\mathcal{M}}}
\newcommand{\cN}{{\mathcal{N}}}
\newcommand{\ck}{{\mathcal{K}}}
\newcommand{\cJ}{{\mathcal{J}}}

\newcommand{\bx}{{\bar{x}}}
\newcommand{\by}{{\bar{y}}}
\newcommand{\bz}{{\bar{z}}}
\newcommand{\bc}{{\bar{c}}}
\newcommand{\OPT}{\textnormal{OPT}}

\newcommand{\eps}{{\varepsilon}}

\newcommand{\floor}[1]{\left\lfloor #1 \right\rfloor}

\DeclareMathOperator*{\argmin}{arg\,min}

\begin{document}
	\newtheorem{thm}{Theorem}[section]
	\newtheorem{prop}[thm]{Proposition}
	\newtheorem{assm}[thm]{Assumption}
	\newtheorem{lem}[thm]{Lemma}
	\newtheorem{obs}[thm]{Observation}
	\newtheorem{cor}[thm]{Corollary}
	\newtheorem{lemma}[thm]{Lemma}
	\newtheorem{definition}[thm]{Definition}
	\newtheorem{theorem}[thm]{Theorem}
	\newtheorem{proposition}[thm]{Proposition}
	\newtheorem{observation}[thm]{Observation}
	\newtheorem{claim}[thm]{Claim}
	\newtheorem{example}[thm]{Example}

	\newtheorem{defn}[thm]{Definition}
	\newcommand{\ariel}[1]{{\color{red} (Ariel :#1)}}
	\newcommand{\ilan}[1]{{\color{blue} Ilan: \color{magenta}#1}}

	\def \II   {{\mathcal I}}
	\newcommand{\one}{\mathbbm{1}}

	\def\claimproof{\proof}
	\def\endclaimproof{\hfill$\square$\\}
	\renewcommand\qedsymbol{$\blacksquare$}

		\title{

			An EPTAS for Budgeted Matroid Independent  Set
		}

		\author{Ilan Doron-Arad\thanks{Computer Science Department, 
				Technion, Haifa, Israel. \texttt{idoron-arad@cs.technion.ac.il}}
			\and
			Ariel Kulik\thanks{CISPA Helmholtz Center for Information Security, Saarland Informatics Campus, Germany. \texttt{ariel.kulik@cispa.de}} 
			\and
			Hadas Shachnai\thanks{Computer Science Department, 
				Technion, Haifa, Israel. \texttt{hadas@cs.technion.ac.il}}
		}
		
		\maketitle

			\begin{abstract}
We consider the {\em budgeted matroid independent set} problem. 
The input is a ground set, where each element has a cost and a non-negative profit, along with a matroid over the elements and a budget.
The goal is to select a subset of elements which maximizes the total profit subject to the matroid and budget constraints. 
Several well known special cases, where we have, e.g., a uniform matroid and a budget, or no matroid constraint (i.e., the classic knapsack problem),
admit a fully polynomial-time approximation scheme (FPTAS).
In contrast, already a slight generalization to the {\em multi-budgeted 
matroid independent set} problem has a PTAS
but does not admit an efficient polynomial-time approximation scheme (EPTAS). This implies a PTAS for our problem, which is the best known result prior to this work.

Our main contribution is an EPTAS for the budgeted matroid independent set problem.
A key idea of the scheme is to find a {\em representative set} for the instance, whose cardinality depends solely on $1/\eps$, 
where $\eps >0$ is the accuracy parameter of the scheme.  The representative set is identified via matroid basis minimization, which can be solved by 
a simple greedy algorithm. Our scheme enumerates over subsets of the representative set  and extends each subset using a linear program. The notion of representative sets may be useful in solving other variants of the budgeted matroid independent set problem. 
\end{abstract}

			\section{Introduction}
\label{sec:introduction}

We consider the {\em budgeted matroid independent set (BMI)} problem, defined as follows.
We are given 
a set of elements $E$, a membership oracle for a collection of independent sets $\II \subseteq 2^{E}$, where $(E, \II)$ is a matroid, a budget $\beta>0$, a cost function~$c:E\to [0,\beta]$, and a profit function $p:E\rightarrow \mathbb{R}_+$. A {\em solution} for the problem is an independent set $X\in \II$ of total cost at most $\beta$ (i.e., $c(X)=\sum_{e\in X}c(e)\leq \beta$). The {\em profit} of a solution $X$ is $p(X)=\sum_{e\in X} p(e)$, and the objective is to find a solution of maximal profit.

BMI is a generalization of the classic {\em$0/1$-knapsack problem}, which is NP-hard
and therefore unlikely to admit an exact polynomial-time algorithm.
Thus, there is a long line of work on finding efficient approximations for
knapsack and its variants (for comprehensive surveys see, e.g.,~\cite{PT98,KPP04}).

Let $\OPT(I)$ be the value of an optimal solution for an instance $I$ of a maximization problem~$\Pi$. 
For $\alpha \in (0,1]$, we say that $\cA$ is an $\alpha$-approximation algorithm
 for $\Pi$ if, for any instance $I$ of $\Pi$, $\cA$ outputs a solution of value at least  $\alpha \OPT(I)$. A {\em polynomial-time approximation scheme} (PTAS)  for a maximization problem $\Pi$ is a family of algorithms $(\cA_{\eps})_{\eps>0}$ such that, for any $\eps>0$, $\cA_{\eps}$ is a polynomial-time $(1-\eps)$-approximation algorithm for $\Pi$; 
 $(\cA_{\eps})_{\eps>0}$ is an EPTAS if the running time of $\cA_{\eps}$ 
 is of the form $f\left(\frac{1}{\eps} \right) \cdot n^{O(1)}$, where $f$ is an arbitrary function, and $n$ is the bit-length encoding size of the input instance;
 $(\cA_{\eps})_{\eps>0}$ is an FPTAS if the running time of $\cA_{\eps}$ is of the form $\left(\frac{n}{\eps}\right)^{O(1)}$.

Polynomial-time approximation schemes allow us to obtain almost optimal solutions 
for NP-hard optimization problems via a speed-accuracy trade-off.
However, the strong dependency of run-times on the error parameter, $\eps >0$, often 
renders these schemes highly impractical. Thus, a natural goal is to seek the fastest
scheme for a given problem, assuming one exists. Having obtained a PTAS, the next step is
to consider possible improvements to an EPTAS, or even to an FPTAS. This is the focus of our work.

For the classic knapsack problem, 
a very efficient FPTAS exists since the 1970's. Lawler's scheme~\cite{La79},
based on ideas
from~\cite{IK75}, achieves a running time of $O(n \log 1/\eps + 1/\eps^4)$ for a $(1-\eps)$-approximation. In contrast, already the {\em two-dimensional knapsack} problem
has a PTAS but does not admit an EPTAS~\cite{KS10}. For the well known {\em multiple knapsack problem}, Chekuri and Khanna~\cite{CK05} derived the first PTAS, which was 
later improved by Jansen to an EPTAS~\cite{jansen2010parameterized, J12}.
The existence of an FPTAS is ruled out by a simple reduction from Partition~\cite{CK05}. More generally,
resolving the complexity status of NP-hard optimization problem with respect to approximation schemes has been  the focus of much research relating to resource allocation and scheduling (see, e.g.,~\cite{W99, DW17,Jansen_et_al:2019, JM19, JSW21} and the surveys ~\cite{SW01,ST18}).

For BMI, FPTASs are known for instances in which the matroid belongs to a restricted family.
One notable example is {\em multiple choice knapsack}~\cite{La79}, where the 
elements are partitioned into groups, and a set is independent if it contains at most one element from each group.
Another example is {\em  $1.5$-dimensional knapsack}~\cite{CKPP00}, in which 
a set is independent if it contains at most $k$ elements for some $ k \geq 1$ (i.e., a  uniform matroid constraint).

There are known results also for generalizations of BMI which involve multiple budget constraints or an additional matroid constraint. Specifically, 
Grandoni and Zenklusen developed in~\cite{GZ10} a  PTAS for  {\em multi-budgeted matroid independent set (MBMI)}, a generalization of BMI in which the costs are $d$-dimensional (for some constant $d\in \mathbb{N}$). The PTAS of \cite{GZ10} is based on integrality properties of a linear programming relaxation of the problem.  
As MBMI generalizes the two-dimensional knapsack problem, 
it does not admit an EPTAS unless $\textnormal{W}[1]=\textnormal{FPT}$  \cite{KS10}. 

The {\em budgeted matroid intersection} problem is a generalization of BMI in which the input includes membership oracles for independent sets of two different matroids, and the solution must be an independent set of both matroids. A PTAS for budgeted matroid intersection was developed in~\cite{BBGS11}. The algorithm of~\cite{BBGS11} uses a Lagrangian relaxation along
with some combinatorial properties of the problem to patch two solutions (i.e., a feasible solution with sub-optimal profit, and a non-feasible solution with high profit) into the final solution. The existence of an EPTAS (or an FPTAS) for budgeted matroid intersection is still open.

The  {\em multi-budgeted matroid intersection} problem is a generalization of both multi-budgeted matroid independent set  and budgeted matroid intersection, in which the cost function is $d$-dimensional, and the input contains two matroids. In \cite{CVZ11} the authors developed a PTAS for multi-budgeted matroid intersection, based on randomized rounding of a solution for a {\em linear programming (LP)} relaxation of the problem.

The budgeted matroid independent set  problem is also a special case of {\em multiple knapsack with a matroid}, a variant of BMI in which the input contains $m$ budgets 
$\beta_1, \ldots , \beta_m$.
A solution consists of $m$ sets $S_1,\ldots,S_m$, where the cost of the $j$th set is at most the $j$th budget ($c(S_j)\leq \beta_j$), and the union of the sets is an independent set of the matroid. A PTAS for the problem (based on randomized rounding) was given in~\cite{FKS21}. The existence of an FPTAS 
is ruled out, as {\em multiple knapsack} is a special case~\cite{CK05}.

To the best of our knowledge, BMI is studied here for the first time. 
A PTAS for the problem follows from known results for any of the above generalizations. 
In all cases, the running time of the 
scheme is dominated by an enumeration phase which guesses the most profitable elements in an optimal solution. 

Our main result is an EPTAS for BMI, thus substantially improving the running times of existing schemes for the problem.
Let $\ck =  (E,\cI,c,p, \beta)$ denote a BMI instance,
 $\OPT(\ck)$ the profit of an optimal solution for $\ck$, and $|\ck|$ the bit-length encoding of the instance $\ck$. 

\begin{theorem}
\label{thm:main}
Given an instance $\ck$ of BMI and $0<\eps<\frac{1}{2}$, there is an algorithm that outputs a solution of profit at least $(1 - \eps)\cdot  \OPT(\ck)$ in time $2^{ O\left(\eps^{-2}\cdot \log \frac{1}{\eps}\right)}  \cdot \textnormal{poly}(|\ck)$.
\end{theorem}

\noindent {\bf Main Technique.} Our algorithm builds upon a framework of Grandoni and Zenklusen~\cite{GZ10} for multi-budgeted matroid independent set. In~\cite{GZ10} the authors show that a basic solution for a linear programming relaxation of the problem has only a few non-integral entries. Thus, a solution for MBMI is constructed by solving the LP relaxation and adding all the elements with non-zero integral entries
to the MBMI solution. An exhaustive enumeration phase which guesses the $\Theta\left(\frac{1}{\eps}\right)$ most profitable elements in an optimal solution is used to mitigate the profit loss caused by discarding the (few) elements with non-integral entries in the solution for the LP.  The running time of the algorithm is dominated by the $|E|^{\Theta(\frac{1}{\eps})}$ operations required for exhaustive enumeration.  

The improved running time of our algorithm is obtained by reducing the 
 time complexity of the enumeration phase. Consider a BMI instance 
 $\ck =  (E,\cI,c,p, \beta)$.
 Given some $0<\eps<\frac{1}{2}$, we say that an element $e\in E$  is {\em profitable} if $p(e) > \eps \cdot \OPT(\ck)$.
 Our algorithm identifies a {\em representative set} of elements
$R\subseteq E$ satisfying $|R|\leq f\left( \frac{1}{\eps}\right)$,
for a computable function $f$. 
Furthermore, we show that the BMI instance has a solution $S$  where all elements in $S\setminus R$ are non-profitable and $p(S) \geq (1-O(\eps)) \cdot \OPT(\ck)$. Thus, we can use enumeration to guess $S\cap R$ and then extend the solution using an LP relaxation similar to~\cite{GZ10}. 
Crucial to our construction of $S$ is the notion of {\em profit gap}, used for identifying elements that may be added to $S$ by solving the LP (see Section~\ref{sec:alg}).
Since all the elements in $S\setminus R$ are non-profitable, the profit loss 
caused by the few non-integral entries is negligible. Moreover, since $|R|\leq f\left( \frac{1}{\eps}\right)$, guessing $S\cap R$ can be done in $2^{f\left(\frac{1}{\eps}\right)}$ steps (in fact, we obtain a slightly better running time), eliminating the dependency of
enumeration on the input size.  The  representative set is identified via matroid basis minimization,
which can be solved by a simple greedy algorithm (see Section~\ref{sec:alg}). 

\noindent {\bf Organization.}
In Section~\ref{sec:preliminaries} we give some definitions and notation. 
Section~\ref{sec:alg} presents our approximation scheme, $\textsf{EPTAS}$, and its analysis.
In Section~\ref{sec:lemMainProof} we give a proof of correctness for algorithm $\textsf{FindRep}$
that is used as a subroutine by the scheme. We conclude in Section~\ref{sec:discussion}
with a summary and open problems.

			\section{Preliminaries}
\label{sec:preliminaries}

For simplicity of the notation, for any set $A$ and an element $e$,  we use $A+e$ and  $A-e$ to denote $A \cup \{e\}$ and $A \setminus \{e\}$, respectively. Also, for any $k \in \mathbb{R}$ let $[k] = \{1,2,\ldots,\floor{k}\}$. Finally, for a function $f:A \rightarrow B$ and a subset of elements $C \subseteq A$, we define $f(C) = \sum_{e \in C} f(e)$. 

\subsection{Matroids}

Let $E$ be a finite ground set and $\cI \subseteq 2^E$ a non-empty set containing subsets of $E$ called the {\em independent sets} of $E$. Then, $\cm = (E, \cI)$ is a {\em matroid} if the following hold: 

\begin{enumerate}
\item (Hereditary Property) For all $A \in \cI$ and $B \subseteq A$, it holds that $B \in \cI$.
	
\item(Exchange Property) For any $A,B \in \cI$ where $|A| > |B|$, there is $e \in A \setminus B$ such that $B +e \in \cI$. 
\end{enumerate} 
The next observation follows 
by repeatedly applying the exchange property.
\begin{observation}
\label{ob:matroid}
Given a matroid $(E,\cI)$ and $A,B \in \cI$, there is $D \subseteq A \setminus B$, $|D| = \max\{|A|-|B|,0\}$ such that $B \cup D \in \cI$. 
\end{observation}

A {\em basis} of a matroid $\cm = (E, \cI)$ is an independent set $B \in \cI$ such that for all $e \in E \setminus B$ it holds that $B+e \notin \cI$. Given a cost function $c:E \rightarrow \mathbb{R}^{+}$, we say that a basis $B$ of $\cm$ is a {\em minimum} basis of $\cm$ w.r.t. $c$ if, for any basis $A$ of $\cm$ it holds that $c(B) \leq c(A)$. A minimum basis of $\cm$ w.r.t. $c$ can be easily constructed in polynomial-time using the greedy approach (see, e.g., ~\cite{cormen2022introduction}). 
In the following we define several operations on matroids that will be useful for deriving our results.

\begin{definition}
\label{def:matroids}
Let $\cm = (E, \cI)$ be a matroid.  
\begin{enumerate}
		
		\item (restriction) For every $F \subseteq E$ define $\cI_{\cap F} = \{A \in \cI~|~ A \subseteq F\}$ and $\cm \cap F = (F, \cI_{\cap F})$.\label{prop1:restriction}
		
		\item (contraction) For every $F \in \cI$ define $\cI_{/ F} = \{A \subseteq E \setminus F~|~ A \cup F \in \cI\}$  and $\cm / F = (E \setminus F, \cI_{/ F})$.\label{prop1:contraction}
		
		\item (truncation) For every $q \in \mathbb{N}$ define $\cI_{\leq q} = \{A \in \cI~|~ |A| \leq q\}$ and $[\cm]_{\leq q} = (E, \cI_{\leq q})$.\label{prop1:trunc}
		
		\item (union) Let $M_1 = (E_1, \cI_1), \ldots, (E_k, \cI_k)$ be matroids; define $\bigvee_{i \in [k]} M_i = (E_{\vee, [k]}, \cI_{\vee, [k]})$, where \newline $E_{\vee, [k]} =  \bigcup_{i \in [k]} E_i$  and $ \cI_{\vee, [k]} =  \left\{\bigcup_{i \in [k]} F_i~\big|~\forall i \in [k]: F_i \in \cI_i\right\}$.
		\label{prop1:union} 	
\end{enumerate}
\end{definition}

The next lemma gathers known results which follow directly
from the definition of a matroid (see, e.g., \cite{schrijver2003combinatorial}). 

\begin{lemma}
\label{lem:prel}
	
Let $\cm = (E, \cI)$ be a matroid. 
\begin{enumerate}
\item 
For any $F \subseteq E$, the restriction of $\cm$ to $F$ (i.e., $\cm \cap F$)  is a matroid.\label{prop:restriction}
		
\item For any $F \in \cI$, the contraction of $\cm$ by $F$ (i.e., $\cm / F$) is a matroid.\label{prop:contraction}
		
\item For any $q \in \mathbb{N}$, the truncation of $\cm$ (i.e., $[\cm]_{\leq q}$) is a matroid.\label{prop:truncation}

\item Given matroids $M_1 = (E_1, \cI_1), \ldots, M_k = (E_k, \cI_k)$, the union of 
$M_1, \ldots, M_k$ (i.e., $\bigvee_{i \in [k]} M_i$) is a matroid.\label{prop:union}
\end{enumerate}
\end{lemma}

\subsection{Matroid polytope}

Let $\cm = (E, \cI)$ be a matroid. Given $B \in \cI$, the {\em indicator vector} of $B$ is the vector $\mathbbm{1}^B \in \{0,1\}^E$, where for all $a \in B$ and $b \in E \setminus B$ we have $\mathbbm{1}^B_a = 1$ and  $\mathbbm{1}^B_b = 0$, respectively. Then the {\em matroid polytope} of $\cm$ is the convex hull of the set of indicator vectors of all independent sets of $\cm$: $P_{\cm} = \textsf{conv} \{\mathbbm{1}^B~|~B \in \cI\}$. 
The next observation will be used in the analysis of our scheme. 
\begin{observation}
	\label{ob:convexHull}
	Let $\cm = (E,\cI)$ be a matroid, and $\bar{x} \in P_{\cm}$. Then $\{e \in E~|~\bar{x}_e = 1\} \in \cI$. 
\end{observation}

			\section{The Algorithm}
\label{sec:alg}

In this section we present an EPTAS for BMI. Our scheme initially handles elements of high profits. This is done by finding a small {\em representative set} out of which the 
scheme selects the most profitable elements in the solution.
 More elements, of lower profits, are then added to the solution using a linear program. For the remainder of this section, fix a BMI instance $\ck =  (E,\cI,c,p,\beta)$ and a parameter $0<\eps <\frac{1}{2}$. W.l.o.g., for all $e \in E$ we assume that 
 $\{e\} \in \cI$ (otherwise $e$ cannot belong to any solution for $\ck$).

Let $H(\ck,\eps) = \{\ell \in E~|~ p(\ell) > \eps \cdot \OPT(\ck)\}$ be the set of {\em profitable} elements in $\ck$, and $E \setminus H(\ck,\eps)$ the set of {\em non-profitable} elements; when understood from the context, we simply use $H = H(\ck,\eps)$. 
We can easily obtain a PTAS by enumerating
over all subsets of at most $\eps^{-1}$ profitable elements to find the profitable elements in the solution. However, such exhaustive search is done in time 
$\Omega\left(|\ck|^{\eps^{-1}}\right)$, which cannot lead to an EPTAS.
Thus, we take a different approach.

A key observation is that efficient solution can be obtained without enumerating over {\em all} subsets of profitable elements. Instead, we limit our scheme to a smaller search space using the notions of {\em replacements} and {\em representative sets}. 
Let $S$ be an independent set having a bounded number of elements. A  {\em replacement} of $S$ 
is a subset of elements which can replace the profitable elements in~$S$, resulting with an independent set of lower cost and almost the same profit.  A representative set $R$ is a subset of elements which  
contains a replacement within $R$ for any independent set with bounded number of elements.   Definitions~\ref{def:Replacement} and~\ref{def:Representatives} give the formal properties of replacements and representative sets, respectively. Let $q(\eps) = \eps^{-\eps^{-1}}$, and recall that $\cI_{\leq q(\eps)} = \{A \in \cI~|~|A| \leq q(\eps)\}$ (the selection of value for $q(\eps)$ 
becomes clear in the proof of Lemma~\ref{lem:Solution}).

 \begin{definition}
	\label{def:Replacement}
	Given a \textnormal{BMI} instance $\ck = (E,\cI,c,p,\beta), 0<\eps<\frac{1}{2}$, $G \in \cI_{\leq q(\eps)}$, and $Z_G \subseteq E$, we say that $Z_G$ is a {\em replacement} of $G$ for $\ck$ and $\eps$ if the following holds: 
	\begin{enumerate}
		
		\item $(G \setminus H) \cup Z_G \in \cI_{\leq q(\eps)}$.\label{p:I}
		\item $c(Z_G) \leq c(G \cap H)$.\label{p:s}
		
		\item $p\left((G \setminus H) \cup Z_G\right) \geq (1-\eps) \cdot p(G)$.\label{p:p}
		
		\item $|Z_G| \leq |G \cap H|$.\label{p:car}
	\end{enumerate}
\end{definition}

 \begin{definition}
\label{def:Representatives}
	Given a \textnormal{BMI} instance $\ck = (E,\cI,c,p,\beta), 0<\eps<\frac{1}{2}$ and $R \subseteq E$, we say that $R$ is a {\em representative set} of $\ck$ and $\eps$ if, for any $G \in \cI_{\leq q(\eps)}$, there is a replacement $Z_G \subseteq R$ of $G$ for $\ck$ and $\eps$. 
\end{definition}

In particular, observe that for any solution $S$ of $\ck$ we have that $S \cap H$ is a replacement of $S$; also, $E$ is a representative set. In the next lemma we show that there exists an almost optimal solution in which all profitable elements belong to a given representative set. Hence, guessing the profitable elements only requires enumerating over subsets of a representative set.
 
\begin{lemma}
	\label{lem:Solution}
	Let $\ck = (E,\cI,c,p,\beta)$ be a \textnormal{BMI} instance and $0<\eps<\frac{1}{2}$. Also, let $R$ be a representatives set of $\ck$ and $\eps$. Then, there is a solution $S$ of $\ck$ such that the following holds.
	
	\begin{enumerate}
		\item $S \cap H \subseteq R$. 
		\item $p\left(S\right) \geq (1-3\eps) \OPT(K)$.
	\end{enumerate} 
\end{lemma}

We give a brief outline of the proof of Lemma~\ref{lem:Solution}. Informally, we consider the elements in some optimal solution, $\OPT$, in non-increasing order by profit. We then partition these elements into three sets: $L,J_{i^*}$ and $Q$, such that the maximal profit in $Q$ is at most $\eps$ times the minimum profit in $L$.
This is the {\em profit gap} of $L$ and $Q$. In addition, $L \in \cI_{\leq q(\eps)}$, and $p(J_{i^*}) \leq \eps \cdot \OPT(\ck)$. Thus, we can use $Z_L \subseteq R$ as a replacement of $L$, i.e., $L \cap H$ will be replaced by elements in $Z_L$ (note that $Z_L$ is not necessarily a subset of the profitable elements).
We now discard $J_{i^*}$, and define $\Delta_L = (L \setminus H) \cup Z_L$. As 
$\Delta_L \cup Q$ may not be an independent set, we use Observation~\ref{ob:matroid}
to construct $T \subseteq Q, |T| \geq |Q|-|L|$  such that $S = \Delta_L \cup T \in \cI$.
An illustration of the proof is given in Figure~\ref{fig:proof}. 

	\begin{figure}[htbp]
	\hspace*{-0.25cm}                                                           
	\includegraphics[scale=0.33]{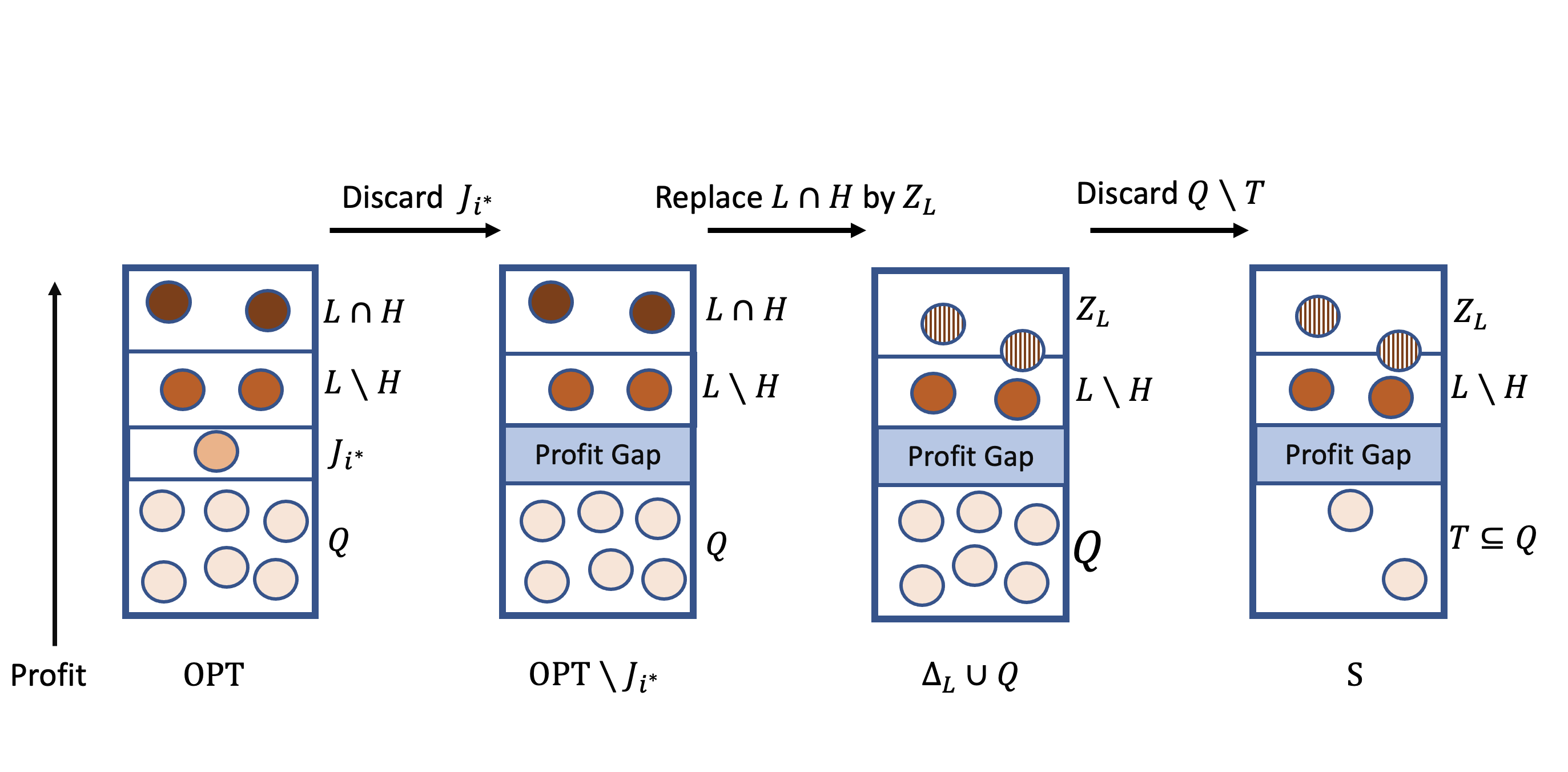}%
	\hspace{1mm}%
	\caption{The construction of the solution $S$ in the proof of Lemma~\ref{lem:Solution}.}
	\label{fig:proof}
\end{figure}

\noindent{\em Proof of Lemma~\ref{lem:Solution}}: 
	With a slight abuse of notation, 
	we use $\OPT$ also to denote the profit of an optimal solution for $\ck$.
	Given an optimal solution, we partition a subset of the elements in the solution 
	  into $\eps^{-1}$ disjoint sets (some sets may be empty). Specifically,
	   let $N = \ceil{\eps^{-1}}$; for all $i \in [N]$ define 
	\begin{equation}
		\label{J_i}
		J_i = \big\{e \in \OPT~\big|~ p(e) \in \big(\eps^{i} \cdot \OPT(\ck),  \eps^{i-1} \cdot \OPT(\ck) \big] \big\}.		
	\end{equation} 
	Let $i^* = \argmin_{i\in [N]} p(J_i)$. By \eqref{J_i} we have at least 
	$\eps^{-1}$ disjoint sets; thus, $p(J_{i^*}) \leq \eps \cdot \OPT(\ck)$. 
Now, let $L~= ~\bigcup_{k \in  [i^*-1]} J_k$ be the subset of all elements in $\OPT$ 
of profits greater than $\eps^{i^*-1} \cdot \OPT(\ck)$, and $Q = \OPT \setminus (L \cup J_{i^*})$. 
To complete the proof of the lemma, we need several claims.
	\begin{claim}
	\label{clam:Iq}
$L \in \cI_{\leq q(\eps)}$. 
\end{claim}
\begin{claimproof}
Since $L \subseteq \OPT$, by the hereditary property of $(E,\cI)$ we have that $L \in \cI$. Also, 
	\begin{equation*}
	\label{contradiction1}
	 |L| \leq \sum_{e \in L} \frac{p(e)}{\eps^{i^*-1} \cdot \OPT(\ck)} = \frac{p(L)}{\eps^{i^*-1} \cdot \OPT(\ck)} \leq  \eps^{- (i^{*}-1)}  \leq   \eps^{-N+1} \leq \eps^{-\eps^{-1}} =  q(\eps)
\end{equation*} The first inequality holds since $p(e) \geq  \eps^{i^*-1} \cdot \OPT$
for all $e \in L$. For the second inequality, we note that $L$ is a solution for $\ck$. By the above and Definition~\ref{def:matroids}, it follows that $L \in \cI_{\leq q(\eps)}$. 
\end{claimproof} 
By Claim~\ref{clam:Iq} and as $R$ is a representative set,
it follows that $L$ has a replacement $Z_L \subseteq R$. 
Let $\Delta_L = (L \setminus H) \cup Z_L$. By Property~\ref{p:I} of Definition~\ref{def:Replacement}, we have that $\Delta_L \in \cI_{\leq q(\eps)}$, and
 by Definition~\ref{def:matroids} it holds that $\cI_{\leq q(\eps)} \subseteq \cI$. 
Hence, $\Delta_L \in \cI$. Furthermore, as $Q \subseteq \OPT \in \cI$, by the 
hereditary property for $(E,\cI)$ we have that $Q \in \cI$. Therefore, by Observation~\ref{ob:matroid}, there is a subset $T \subseteq Q \setminus\Delta_L$, where $|T|  = \max\{|Q| - |\Delta_L|,0\}$, such that $\Delta_L \cup T \in \cI$. 

Let $S = \Delta_L \cup T$. We show that $S$ satisfies the conditions of the lemma. 

\begin{claim}
\label{claim:IsSolution}
 $S$ is a solution for $\ck$. 
\end{claim}

\begin{claimproof}
By the definition of $T$ it holds that $S = \Delta_L \cup T \in \cI$. Moreover, 	
\[
\begin{array}{ll}
c(S) &= c(\Delta_L\cup T) \\
& \leq c(Z_L)+ c(L \setminus H) +c(T) \\
& \leq c(L \cap H)+c(L \setminus H)+c(T) \leq c(L)+c(Q) \\ 
& \leq c(\OPT)\\
& \leq \beta.
\end{array}
\]
The second inequality holds since $c(Z_L) \leq c(L \cap H)$ (see Property~\ref{p:s} of Definition~\ref{def:Replacement}). For the third inequality, recall that $T \subseteq Q$. The last inequality holds since $\OPT$ is a solution for $\ck$. 
\end{claimproof} 

The proof of the next claim relies on
the profit gap between the elements in $Q$ and $L$. 
	\begin{claim}
	\label{clam:profitBound1}
	$p\left(Q \setminus T\right) \leq \eps \cdot \OPT(\ck)$.
\end{claim}
\begin{claimproof} 

	Observe that 
	\begin{equation}
		\label{eq:menS}
		|Q \setminus T| \leq |\Delta_L| \leq |Z_L|+|L \setminus H| \leq |L \cap H|+|L \setminus H| \leq |L|.
	\end{equation} 
	The first inequality follows from the definition of $T$. For the third inequality, we use Property~\ref{p:car} of Definition~\ref{def:Replacement}. Hence, 
	\begin{equation*}
		\begin{aligned}
			p(Q \setminus T) \leq{} & |Q \setminus T| \cdot \eps^{i^*} \cdot \OPT(\ck) \leq  |L| \cdot \eps^{i^*} \cdot \OPT(\ck) \leq \eps \cdot p(L) \leq \eps \cdot \OPT(\ck).
		\end{aligned}
	\end{equation*} 
	
	The first inequality holds since $p(e) \leq \eps^{i^*} \cdot \OPT(\ck)$ for all $e \in Q$. 
	The second inequality is by \eqref{eq:menS}. The third inequality holds since $p(e) > \eps^{i^*-1} \cdot \OPT(\ck)$ for all $e \in L$. 
\end{claimproof} 	

\begin{claim}
	\label{clam:profitBound2}
	$p\left(S\right) \geq (1-3\eps) \OPT(\ck)$.
\end{claim} 

\begin{claimproof} 
	By Property~\ref{p:p} of Definition~\ref{def:Replacement},
	\begin{equation}
			\label{eq:proofProfit1}
		\begin{aligned}
			 p(\Delta_L) ={} & p((L \setminus H) \cup Z_L) \geq (1-\eps) \cdot p(L).
		\end{aligned}
	\end{equation} Moreover, 
		\begin{equation}
			\label{eq:proofProfit2}
		\begin{aligned}
			 p(T) {} & \geq p(Q) - p(Q \setminus T) \geq p(Q) - \eps \cdot \OPT(\ck) \geq (1-\eps) \cdot p(Q) - \eps \cdot \OPT(\ck). 
		\end{aligned}
	\end{equation} 
	
	The second inequality is by Claim~\ref{clam:profitBound1}. Using \eqref{eq:proofProfit1} and \eqref{eq:proofProfit2}, we have that 
	
\[	
\begin{array}{ll}
p(\Delta_L)+p(T) & \geq (1-\eps) \cdot p(L \cup Q)-\eps \cdot \OPT(\ck) 
\\ & = (1-\eps) \cdot p(\OPT \setminus J_{i^*})-\eps \cdot \OPT(\ck) \\
&  \geq (1-3\eps) \cdot \OPT(\ck). 
\end{array}
\]

The last inequality holds since $p(J_{i^*}) \leq \eps \cdot \OPT(\ck)$. Observe that 
$S = \Delta_L \cup T$ and $T \cap \Delta_L = \emptyset$. Therefore,  $$p(S) = p(\Delta_L)+p(T) \geq (1-3\eps) \cdot \OPT(\ck).$$

\end{claimproof} 
 Finally, we note that, by \eqref{J_i} and the definition of $Q$, $Q \cap H = \emptyset$. Consequently, by the definition of $S$, we have that $S \cap H \subseteq Z_L$. 
 As $Z_L \subseteq R$ (by Definition~\ref{def:Representatives}), it follows that $S \cap H \subseteq R$. Hence, using Claims~\ref{claim:IsSolution} and~\ref{clam:profitBound2},
 we have the statement of the lemma.  \qed

Our scheme for BMI constructs a representative set whose cardinality depends solely on $\eps$. To this end, we first partition the profitable elements (and possibly some more elements) into a small number of {\em profit classes}, where elements from the same profit class have {\em similar} profits. 
The profit classes are derived from a $2$-approximation $\alpha$ for $\OPT(\ck)$, which can be easily computed in polynomial time. specifically, for all $r \in [\log_{1-\eps} \left(\frac{\eps}{2}\right)+1]$ define the $r$-{\em profit class} as \begin{equation}
	\label{Er}
	C_{r}(\alpha) = \left\{e \in E~\bigg|~ \frac{ p(e)}{2 \cdot \alpha} \in \big( (1-\eps)^{r}, (1-\eps)^{r-1} \big]\right\}.
\end{equation} 

For each  
$r \in [\log_{1-\eps} \left(\frac{\eps}{2}\right)+1]$, we define $\left[ (E, \cI) \cap C_{r}(\alpha) \right]_{\leq q(\eps)}$, the {\em corresponding matroid} for the
$r$-profit class.
We construct a representative set by computing a minimum basis w.r.t. the cost function $c$, for the matroid defined as the union of the corresponding matroids of all profit classes. Note that by Lemma~\ref{lem:prel} the latter is a matroid.  
The pseudocode of algorithm $\textsf{FindRep}$, which outputs a representative set, is given in Algorithm~\ref{alg:representative}.

\begin{algorithm}[h]
	\caption{$\textsf{FindRep}(\ck = (E,\cI,c,p,\beta),\eps)$}
	\label{alg:representative}
	
	Compute a $2$-approximation $S^*$ for $\ck$ using a PTAS for BMI with parameter $\eps' = \frac{1}{2}$.\label{step:APP}
	
	 Set $\alpha = p(S^*)$. \label{step:alpha}
	
		Return a minimum basis w.r.t. $c$ of the matroid $$\bigvee_{r \in [\log_{1-\eps} \left(\frac{\eps}{2}\right)+1]} \left[ (E, \cI) \cap C_{r}(\alpha) \right]_{\leq q(\eps)}.$$\label{step:Br}
\end{algorithm}

\begin{lemma}
	\label{lem:main}
	Given a \textnormal{BMI} instance $\ck = (E,\cI,c,p,\beta)$ and $0<\eps <\frac{1}{2}$, Algorithm~\ref{alg:representative} returns in time $\textnormal{poly}(|\ck|)$ 
	 a representative set $R$ of $\ck$ and $\eps$, such that $|R| \leq q(\eps) \cdot \left(\log_{1-\eps} \left(\frac{\eps}{2}\right)+1\right) $.
\end{lemma} 

In the following we outline the main ideas of the proof.
Let $R = \textsf{FindRep}(\ck,\eps)$, and consider some $G \in \cI_{\leq q(\eps)}$ and $e \in (G \cap H) \setminus R$.  By \eqref{Er}  the element $e$ belongs to some profit class $C_r(\alpha), r \in [\log_{1-\eps} \left(\frac{\eps}{2}\right)+1]$. Since $R$ is a minimum basis w.r.t. $c$, we can use matroid properties to show that there is some $b \in C_r(\alpha) \cap R$ of cost $c(b) \leq c(e)$ such that $G-a+b \in \cI_{\leq q(\eps)}$. 
We now keep replacing elements in $(G \cap H) \setminus R$ in a similar manner, until no such element exists. Thus, we have a replacement of $G$ within $R$, i.e.,
$R$ is a representative set. A formal proof of Lemma~\ref{lem:main} is given in Section~\ref{sec:lemMainProof}.

Our scheme uses the representative set to select profitable elements for the solution. Using a linear program, the solution is extended to include also non-profitable elements.
As the exact set of non-profitable elements is unknown, we use an approximation for the optimal profit. Specifically, let $\frac{\OPT(\ck)}{2} \leq \alpha \leq \OPT(\ck)$, and
denote by $E(\alpha) = \{e \in E~|~ p(e) \leq 2\eps \cdot \alpha\}$ the set 
including the non-profitable elements, and possibly also profitable elements $e$ 
such that $p(e) \leq 2 \eps \cdot \OPT(\ck)$. 	

The LP is based on the matroid polytope of the following matroid. Given a solution $F$ for $\ck$, we define $\cm_F (\alpha) = \left( (E,\cI)/F \right)  \cap E(\alpha)  $; in other words,  $\cm_F(\alpha) = (E(\alpha),\cI_F(\alpha))$, where  $\cI_F(\alpha) = \{A \subseteq E(\alpha)~|~ A \cup F \in \cI\}$. Note that $\cm_{F}(\alpha)$ is indeed a matroid, by Properties~\ref{prop:restriction} and~\ref{prop:contraction} of Lemma~\ref{lem:prel}.
The LP formulation is given by 
	\begin{equation}
		\label{LP}
		\begin{aligned}
			\textnormal{LP}(\ck,F,\alpha)~~~~~~~~~~~~~~~~~ \max\quad        & ~~~\sum_{e \in E(\alpha) \setminus F} \bar{x}_{e}   \cdot p(e)             ~~~~~~~~~~~~~~~~~~~~~                                                \\
			\textsf{s.t.\quad} & ~~~\sum_{e \in E(\alpha) \setminus F} \bar{x}_{e}  \cdot c(e)         \leq  \beta-c(F) ~~~~~~~~~~~~~~~~~~~~~           \\  
			& ~~~\bar{x} \in P_{\cm_F (\alpha)} \\
		\end{aligned}
	\end{equation}  
The linear program $\textnormal{LP}(\ck,F,\alpha)$  maximizes the total profit of a point in the matroid polytope of $\cm_{F} (\alpha)$ (i.e., $P_{\cm_F(\alpha)}$), such that the total cost of elements is at most $\beta-c(F)$; that is, the residual budget 
after selecting for the solution the elements in $F$. 
\begin{observation}
	\label{ob:LP}
	Let $\ck = (E,\cI,c,p,\beta)$ be a \textnormal{BMI} instance, $\frac{\OPT(\ck)}{2} \leq \alpha \leq \OPT(\ck)$, $S$ a solution for $\ck$, and $\bar{x}$ an optimal  basic solution for $\textnormal{LP}(\ck,S \cap H,\alpha)$. Then, $\sum_{e \in E(\alpha) \setminus (S \cap H)} \bar{x}_{e}   \cdot p(e) \geq p\left(S \setminus H \right)$. 
\end{observation}

It is folklore that a linear program such as $\textnormal{LP}(\ck,F,\alpha)$ can be solved in polynomial-time in $|\ck|$. 
As we could not find a proper reference, we include the proof of the next lemma in the appendix. 
\begin{lemma}
	\label{lem:LPrunning}
	For any   \textnormal{BMI} instance $\ck = (E,\cI,c,p,\beta)$, $\frac{\OPT(\ck)}{2} \leq \alpha \leq \OPT(\ck)$, and a solution $F$ of $\ck$, a basic optimal solution of $\textnormal{LP}(\ck,F,\alpha)$ can be found in time $\textnormal{poly}(|\ck|)$. 
\end{lemma}
The next lemma will be 
useful for deriving a solution of high profit using $\textnormal{LP}(\ck,F,\alpha)$.
The proof follows as a special case of a result of \cite{GZ10}.

 \begin{lemma}
		\label{lem:integral}
		Let $\ck = (E,\cI,c,p,\beta)$ be a \textnormal{BMI} instance, $\frac{\OPT(\ck)}{2} \leq \alpha \leq \OPT(\ck)$, a solution $F$ of $\ck$, and $\bar{x}$  a basic solution for $\textnormal{LP}(\ck,F,\alpha)$. Then $\bar{x}$ has at most two non-integral entries. 
	\end{lemma}
	
Using the above, we have the required components for an EPTAS for BMI. Let $R$ be the representative set returned by $\textsf{FindRep}(\ck,\eps)$. For all solutions $F \subseteq R$ with $|F| \leq \eps^{-1}$, we find a basic optimal solution $\bar{\lambda}^F$ for $\textnormal{LP}(\ck,F,\alpha)$ and define $C_F =  \{e \in E(\alpha) \setminus F~|~ \bar{\lambda}^F_{e} = 1\} \cup F$ as the  {\em solution} of $F$. 
Our scheme iterates over the solutions $C_F$ for all such subsets $F$ and chooses
a solution $C_{F^*}$ of maximal total profit.
The pseudocode of the scheme is given in Algorithm~\ref{alg:EPTAS}.	
	
	\begin{algorithm}[h]
		\caption{$\textsf{EPTAS}(\ck = (E,\cI,c,p,\beta),\eps)$}
		\label{alg:EPTAS}
		
		
		Construct the representative elements $R = \textsf{FindRep}(\ck,\eps)$.\label{step:rep}
		
		Compute a $2$-approximation $S^*$ for $\ck$ using a PTAS for BMI with parameter $\eps' = \frac{1}{2}$.\label{step:APP2}
		
		Set $\alpha = p(S^*)$.
		
		Initialize an empty solution $A \leftarrow \emptyset$.\label{step:init}
		
		\For{$F \subseteq R \textnormal{ s.t. } |F| \leq \eps^{-1} \textnormal{ and } F \textnormal{ is a solution of } \ck $ \label{step:for}}{
			

			Find a basic optimal solution $\bar{\lambda}^F$ of $\textnormal{LP}(\ck,F,\alpha)$.\label{step:vertex}
			
			Let $C_F =  \{e \in E(\alpha) \setminus F~|~ \bar{\lambda}^F_{e} = 1\} \cup F$.\label{step:Cf}

			\If{$p\left(C_F\right) > p(A)$\label{step:iff}}{
				
				Update $A \leftarrow C_F$\label{step:update}
				
			}

		}
		
		Return $A$.\label{step:retA}
	\end{algorithm}

	\begin{lemma}
		\label{thm:EPTAS}
		Given a \textnormal{BMI} instance $\ck = (E,\cI,c,p,\beta)$ and $0<\eps<\frac{1}{2}$, Algorithm~\ref{alg:EPTAS} returns a solution for $\ck$ of profit at least $(1-7\eps) \cdot \OPT(\ck)$. 
	\end{lemma}
	\begin{proof}
		By Lemma~\ref{lem:Solution},  there is a solution $S$ for $\ck$ such that $S \cap H \subseteq R$, and $p\left(S\right) \geq (1-3\eps) \OPT(K)$. As for all $e \in S \cap H$ we have $p(e) > \eps \cdot \OPT(\ck)$, and $S$ is a solution for $\ck$, it follows that $|S \cap H| \leq \eps^{-1}$.
		We note that there is an iteration of Step~\ref{step:for} in which $F = S \cap H$; thus, in Step~\ref{step:vertex} we construct a basic optimal solution $\bar{\lambda}^{S \cap H}$ of $\textnormal{LP}(\ck,S \cap H,\alpha)$. We use $X(A) = \{e \in E(\alpha) \setminus A~|~ \bar{\lambda}^{A}_e = 1\}$ for every basic solution $\bar{\lambda}^{A}$ computed in Step~\ref{step:vertex} for $A \subseteq R$.  
		Then,
		 \begin{equation}
			\label{eq:finalProfitA}
			\begin{aligned}
				p\left(X(S \cap H)\right) ={} & \sum_{e \in E(\alpha) \setminus (S \cap H)~\text{s.t.}~\bar{\lambda}^{S \cap H}_{e} = 1} p(e) \\ 
				\geq{} & \sum_{e \in E(\alpha) \setminus (S \cap H)} \bar{\lambda}^{S \cap H}_{e}   \cdot p(e) - 2 \cdot 2\eps \cdot \OPT(\ck) \\ \geq{} &  p(S\setminus H) -4\eps \cdot \OPT(\ck).
			\end{aligned}
		\end{equation}
		 The first inequality holds since, by Lemma~\ref{lem:integral}, $|\{e \in E(\alpha) \setminus (S \cap H)~|~\bar{\lambda}^{S \cap H}_{e} \in (0,1)\}| \leq 2$, and for all $e \in E(\alpha)$, $p(e) \leq 2\eps \cdot \alpha \leq 2 \eps \cdot \OPT(\ck)$. The second inequality follows from Observation~\ref{ob:LP}. Now, 
		 \begin{equation}
			\label{eq:finalProfit}
			\begin{aligned}
					p(C_{S \cap H}) ={} & p(S \cap H)+p\left(X(S \cap H)\right)  \geq  p(S)- 4\eps \cdot \OPT(\ck) \geq (1-7\eps) \OPT(K).
			\end{aligned}
		\end{equation}
		
		 The first inequality uses~\eqref{eq:finalProfitA}. The last inequality is by Lemma~\ref{lem:Solution}.   \begin{claim}
		\label{claim:Cf}
$A  = \textnormal{\textsf{EPTAS}}(\ck,\eps)$ is a solution of $\ck$. 
	\end{claim}
\begin{claimproof} If $A = \emptyset$ the claim trivially follows since $\emptyset$ is a solution of $\ck$. Otherwise, by Step~\ref{step:update} of Algorithm~\ref{alg:EPTAS}, there is a solution $F$ of $\ck$ such that $A = C_F$. By Observation~\ref{ob:convexHull}, $X(F)$ is an independent set in the matroid $M_{F}(\alpha)$. 
Also, recall that $\cm_{F} (\alpha) = \left((E,\cI) /F \right) \cap E(\alpha) $. Hence, by Definition~\ref{def:matroids}, we have that $X(F) \cup F =C_F \in \cI$. Additionally, 
 $$c\left(C_{F}\right) = c(F)+\sum_{e \in X(F)} c(e) \leq c(F)+\beta-c(F) = \beta.$$ 
 The inequality follows from~\eqref{LP}. By the above, we conclude that $A$ is a solution for $\ck$.  
\end{claimproof} 
	
By Claim~\ref{claim:Cf}, Steps~\ref{step:for}, \ref{step:update} and~\ref{step:retA}
of Algorithm~\ref{alg:EPTAS} and \eqref{eq:finalProfit}, we have that
$A = \textsf{EPTAS}(\ck,\eps)$ is a solution for $\ck$ satisfying $p(A) \geq p(C_{S \cap H}) \geq (1-7\eps) \OPT(K)$. This completes the proof.

	\end{proof}
	\begin{lemma}
	\label{thm:running}
	Given a \textnormal{BMI} instance $\ck = (E,\cI,c,p,\beta)$ and $0<\eps<\frac{1}{2}$, Algorithm~\ref{alg:EPTAS} runs in time %
	$2^{ O \left(\eps^{-2} \log \frac{1}{\eps} \right)} \cdot \textnormal{poly}(|\ck|)$. 
\end{lemma}
\begin{proof}
By Lemma~\ref{lem:main}, the time complexity of Step~\ref{step:rep} 
is $\textnormal{poly}(|\ck|)$. Step~\ref{step:APP2} can be computed in time $\textnormal{poly}(|\ck|)$, by using a PTAS for BMI taking 
$\eps=\frac{1}{2}$ (see, e.g.,~\cite{GZ10}). Now, using logarithm rules we have
\begin{equation}
	\label{eq:log1}
  \log_{1-\eps} \left(\frac{\eps}{2}\right)+1  \leq 
  \frac{\ln \left(\frac{2}{\eps}\right)}{-\ln \left(1-\eps \right)}+1 \leq \frac{2\eps^{-1}}{\eps}+1 \leq 3 \eps^{-2}. 
\end{equation} 
The second inequality follows from  $x< -\ln (1-x), \forall x>-1, x \neq 0$, and $\ln y < y, \forall y>0$. Therefore, \begin{equation}
\label{eq:R}
|R| \leq \left(\log_{1-\eps} \left(\frac{\eps}{2}\right)+1\right)  \cdot q(\eps) \leq 3\eps^{-2} \cdot \eps^{-\eps^{-1}} \leq \eps^{-3\eps^{-1}}. 
\end{equation}
 The first inequality is by Lemma~\ref{lem:main}. The second inequality is by \eqref{eq:log1}. Let $$W = 	\big\{F \subseteq R~\big|~ |F| \leq \eps^{-1}, F \in \cI, c(F) \leq \beta\big\}$$ be the set of solutions considered in Step~\ref{step:for} of Algorithm~\ref{alg:EPTAS}. Then,
 \begin{equation}
\label{eq:subR}
\begin{aligned}
|W| \leq{} &  \left(|R|+1\right)^{\eps^{-1}}
	\leq  {\left(\eps^{-3\eps^{-1}}+1\right)}^{\eps^{-1}}
		 \leq {\left(\eps^{-4\eps^{-2}}\right)} = 2^{ O \left(\eps^{-2} \log \frac{1}{\eps} \right)}.
\end{aligned}
\end{equation} The second inequality is by \eqref{eq:R}. Hence, by \eqref{eq:subR}, the number of iterations of the {\bf for} loop  in Step~\ref{step:for} is bounded by $2^{ O \left(\eps^{-2} \log \frac{1}{\eps} \right)}$. In addition, by Lemma~\ref{lem:LPrunning}, the running time of each iteration is  $\textnormal{poly}(|\ck|)$. By the above, the running time of Algorithm~\ref{alg:EPTAS} is $2^{ O \left(\eps^{-2} \log \frac{1}{\eps} \right)} \cdot \textnormal{poly}(|\ck|)$.  
\end{proof}

\noindent{\bf Proof of Theorem~\ref{thm:main}:} Given a BMI instance $\ck$ and $0<\eps<\frac{1}{2}$, using Algorithm~\ref{alg:EPTAS} for $\ck$ with parameter $\frac{\eps}{7}$ we have the desired approximation guarantee by Lemma~\ref{thm:EPTAS}. Also, by Lemma~\ref{thm:running}, the running time is $2^{ O \left(\eps^{-2} \log \frac{1}{\eps} \right)} \cdot \textnormal{poly}(|\ck|)$. \qed

			\section{Correctness of  \textsf{FindRep}}
\label{sec:lemMainProof}
In this section we give the proof of Lemma~\ref{lem:main}. The proof 
is based on {\em substitution} of subsets by profit classes. 
A substitution is closely related to 
replacement (see Definition~\ref{def:Replacement}); indeed, we can construct replacements for independent sets (and therefore a representative set) by specific substitutions. 
We start by stating several lemmas that will be used in the proof of Lemma~\ref{lem:main}.
The first lemma refers to a {\em generalized exchange property} of matroids. 
\begin{lemma}
	\label{lem:gen}
	Let $\cm = (E,\cI)$ be a matroid, $A,B \in \cI$, and $a \in A \setminus B$ such that $B+a \notin \cI$. Then there is $b \in B \setminus A$ such that $A-a+b \in \cI$. 
\end{lemma}
\begin{proof}
	By Observation~\ref{ob:matroid}, there is $C \subseteq A \setminus B, |C| = \max\{|A|-|B|, 0\}$ such that $B \cup C \in \cI$. Then, 
	\begin{equation}
		\label{eq:>}
		|B \cup C| \geq |B|+|A|-|B| > |A|-1 = |A-a|.
	\end{equation}
	
 The first inequality holds since $|C| = \max\{|A|-|B|, 0\}$ and $C \cap B = \emptyset$. 
 Recall that $B \cup C \in \cI$; also, by  the hereditary property of $\cm$, $A-a \in \cI$. Therefore, by \eqref{eq:>} and the exchange property of $\cm$, there is $b \in (B \cup C) \setminus (A-a)$ such that $A-a+b \in \cI$. As $C \subseteq A$, it follows that $(B \cup C) \setminus (A-a) \subseteq (B \setminus A) +a$, thus $b \in (B \setminus A) +a$. Now, assume towards contradiction that $b = a$; then, $a \in C$ since $b \in (B \cup C) \setminus (A-a)$ and $a \in A \setminus B$. As $B \cup C \in \cI$, by the hereditary property of $\cm$ it follows that $B+a \in \cI$. Contradiction. Hence, $b \neq a$ and it follows that $b \in B \setminus A$ and $A-a+b \in \cI$.  
\end{proof}

The next lemma gives a general property of minimum bases of matroids.
\begin{lemma}
\label{lem:notIS}
Given a matroid $\cm = (E,\cI)$ and a weight function $w:E \rightarrow \mathbb{R}_{\geq 0}$, let $B$ be a minimum basis of $\cm$ w.r.t. $w$. Then, for any $a \in E \setminus B$ it holds that $\{e\in B ~|~ w(e)\leq w(a)\}+ a \notin \cI$. 	 
\end{lemma}

\begin{proof}
	Let $D = \{e\in B ~|~ w(e)\leq w(a)\}$. Assume towards contradiction that $D + a \in \cI$. Then, by Observation~\ref{ob:matroid} there is $C \in B \setminus (D+a), |C| = \max \{|B|-|D+a|,0\}$ such that $ (D+a)\cup C \in \cI$; let $F =  (D+a)\cup C$. Note that $|F| = |D+a|+|C| \geq |D+a|+|B|-|D+a| =  |B|$. Hence, as $|F| \geq |B|$ and $B$ is a basis of $\cm$, we have that $F$ is a basis of $\cm$ (since all bases have the same cardinality; see, e.g., (39.2) in~\cite{schrijver2003combinatorial}). Noting that  $a \notin B$, $F -a \subseteq B$ and $|F| = |B|$, there is $e \in B \setminus F$ such that $F = B+a-e$. As $e \in B \setminus F$ and  $B \setminus F \subseteq B \setminus D$  it holds that $e \in B \setminus D$, and it follows that $w(a)<w(e)$. Therefore,
	  \begin{equation}
		\label{eq:contradictionD}
		w(F) = w((D+a) \cup C) =  w(B)+w(a)-w(e) < w(B)
	\end{equation} 
	By~\eqref{eq:contradictionD}, we have that $F$ is a basis of $\cm$ satisfying
	$w(F)<w(B)$. Contradiction (to the minimality of $B$ w.r.t. $w$).
\end{proof}
In the next lemma we consider a simple property of bases of union matroids. 
\begin{lemma}
	\label{lem:minBasisGeneral}
	 Let $M_1 = (E_1, \cI_1), \ldots, M_k=(E_k, \cI_k)$ be matroids such that $E_i \cap E_j  = \emptyset ~\forall i,j \in [k], i \neq j$; also, let $w:E \rightarrow \mathbb{R}_{\geq 0}$ be a weight function, and $R$ a minimum basis of $\bigvee_{i \in [k]} M_i$ w.r.t.~$w$. Then, for all $i \in [k]$, $R \cap E_i$ is a  minimum basis of $M_i$ w.r.t.~$w$.\footnote{With a slight abuse of notation we use $w$ also for the restriction of $w$ on $E_i$.}
\end{lemma}
\begin{proof}
	Let $i \in [k]$.  For all $e \in E_i \setminus R$, it holds that $R \cap E_i+e$ is not an independent set of $M_i$ since otherwise $R+e$ is an independent set of $\bigvee_{i \in [k]} M_i$ by Definition~\ref{def:matroids},  contradicting that $R$ is a basis of $\bigvee_{i \in [k]} M_i$; we conclude that $R \cap E_i$ is a basis of  $M_i$. Assume towards contradiction that there is a basis $B$ of the matroid $M_i$ such that $w(B)< w(R \cap E_i)$. As $E_i \cap E_j  = \emptyset ~\forall i,j \in [k], i \neq j$, by Definition~\ref{def:matroids} it follows that $\left(R \setminus E_i\right) \cup B$ is a basis of the matroid $\bigvee_{i \in [k]} M_i$. In addition, $$w(\left(R \setminus E_i\right) \cup B) = w(R)-w(R \cap E_i)+w(B) < w(R).$$ We reach a contradiction since $R$ is a minimum basis w.r.t. $w$ of $\bigvee_{i \in [k]} M_i$. 
\end{proof}

We now prove Lemma~\ref{lem:main} using several auxiliary lemmas. For the remainder of this section, let $\ck = (E,\cI,c,p,\beta)$ be a BMI instance, $0<\eps<\frac{1}{2}$, $R = \textsf{FindRep}(\ck,\eps)$, and $\alpha$ the value from Step~\ref{step:alpha} of Algorithm~\ref{alg:representative}. For the next lemma, recall the sets $C_{r}(\alpha)$ for  $r \in [\log_{1-\eps} \left(\frac{\eps}{2}\right)+1]$ were  defined in \eqref{Er}. 
\begin{lemma}
	\label{lem:minBasis}
	For all $r \in [\log_{1-\eps} \left(\frac{\eps}{2}\right)+1]$ it holds that $R \cap C_r(\alpha)$ is a minimum basis w.r.t. $c$ of the matroid $\cm =  \left[ (E, \cI) \cap C_{r}(\alpha) \right]_{\leq q(\eps)}$, where $R = \textsf{FindRep}(\ck,\eps)$. 
\end{lemma}

\begin{proof}
	Observe that for all $r,t \in [\log_{1-\eps} \left(\frac{\eps}{2}\right)+1]$ it holds that $C_r(\alpha) \cap C_t(\alpha) = \emptyset$ by \eqref{Er}. Hence, the statement of the lemma follows from Step~\ref{step:Br} of Algorithm~\ref{alg:representative} and Lemma~\ref{lem:minBasisGeneral}. 
\end{proof}
\begin{lemma}
	\label{claim:RR}
	For all $e \in H$ there is exactly one $r \in [\log_{1-\eps} \left(\frac{\eps}{2}\right)+1]$ such that $e \in C_r(\alpha)$.
\end{lemma}
\begin{proof}
	Let $e \in H$. Observe that: \begin{equation}
		\label{eq:app}
		\frac{\eps}{2} \leq \frac{p(e)}{2 \cdot \OPT(\ck)} \leq \frac{p(e)}{2 \cdot \alpha} \leq \frac{p(e)}{\OPT(\ck)} \leq 1.
	\end{equation}The first inequality holds since $e \in H$. The second and third inequalities are because $\alpha$ is the value of a $2$-approximation for the optimum of $\ck$; thus, $\frac{\OPT(\ck)}{2} \leq \alpha \leq \OPT(\ck)$. The last inequality is because $\{e\}$ is a solution of $\ck$. In addition, observe that for $r_0 = \ceil{\log_{1-\eps} \frac{\eps}{2}}$ it holds that $(1-\eps)^{r_0} \leq \frac{\eps}{2}$ and that for $r_1 = 1$ it holds that $(1-\eps)^{r_1-1} = 1$. Hence, because $r_0 \leq \floor{\log_{1-\eps} \left(\frac{\eps}{2}\right)+1}$, by \eqref{eq:app}, there is exactly one $r \in [\log_{1-\eps} \left(\frac{\eps}{2}\right)+1]$ such that $\frac{p(e)}{2\alpha} \in \big((1-\eps)^r, (1-\eps)^{r-1}\big]$; thus, by \eqref{Er} it holds that $e \in C_r(\alpha)$ and $e \notin C_{r'}(\alpha)$ for $r' \in [\log_{1-\eps} \left(\frac{\eps}{2}\right)+1] \setminus \{r\}$.  
\end{proof}

For the proof of Lemma~\ref{lem:main}, we define a {\em substitution} of some independent set. The first two properties of substitution of some $G \in \cI_{\leq q(\eps)}$ are identical to the first two properties in the definition of a replacement of $G$. 
However, we require that a substitution preserves the number of profitable elements in $G$ from each profit class, and that a substitution must be disjoint to the set of non-profitable elements in $G$. %

\begin{definition}
	\label{def:sub}
	For $G \in \cI_{\leq q(\eps)}$ and $Z_G \subseteq \bigcup_{r \in [\log_{1-\eps} \left(\frac{\eps}{2}\right)+1] } C_r(\alpha)$, we say that $Z_G$ is a {\em substitution} of $G$ if the following holds. \begin{enumerate}

		\item $(G \setminus H) \cup Z_G \in \cI_{\leq q(\eps)}$.\label{pp:1}
		
	\item  $c(Z_G) \leq c(G \cap H)$.\label{pp:2}
	
	\item For all $r \in [\log_{1-\eps} \left(\frac{\eps}{2}\right)+1]$ it holds that $|C_r(\alpha) \cap Z_G| = |C_r(\alpha) \cap G \cap H|$.\label{pp:3}

	\item $(G \setminus H) \cap Z_G = \emptyset$.\label{pp:4}
	\end{enumerate}
\end{definition}

\begin{lemma}
	\label{lem:R1}
	For all $G \in \cI_{\leq q(\eps)}$ there is a substitution $Z_G$ of $G$ such that $Z_G \subseteq R$.

\end{lemma}%
In the proof of Lemma~\ref{lem:R1}, we assume towards a contradiction that there is no substitution for some $G \in \cI_{\leq q(\eps)}$ which is a subset of $R$. To reach a contradiction, we take a substitution $Z_G$ of $G$ with maximal number of elements from $R$ and show that using Lemma~\ref{lem:gen} and Lemma~\ref{lem:notIS} we can create a substitution with more elements from $R$ by replacing an element from $Z_G \setminus R$ by an element from $R \setminus Z_G$.  
\begin{proof}[Proof of Lemma~\ref{lem:R1}]
Let $G \in \cI_{\leq q(\eps)}$ and let $Z_G$ be a substitution of $G$ such that $|Z_G \cap R|$ is maximal among all substitutions of $G$; formally, let $\mathcal{S}(G)$ be all substitutions of $G$ and let $Z_G \in \{ Z \in \mathcal{S}(G)~|~ |Z \cap R| = \max_{Z' \in \mathcal{S}(G)} |Z' \cap R|\}$. Since $G \cap H$ is in particular a substitution of $G$  it follows that $\mathcal{S}(G)\neq \emptyset$, and thus  $Z_G$ is well defined.  Assume towards a contradiction that there is $a \in Z_G \setminus R$; then, by Definition~\ref{def:sub} there is $r \in [\log_{1-\eps} \left(\frac{\eps}{2}\right)+1]$ such that $a \in C_r(\alpha)$. Let $\Delta_G = (G \setminus H) \cup Z_G$. 
\begin{claim}
	\label{claim:auxRep}
	There is $b \in (C_r(\alpha) \cap R) \setminus \Delta_G$ such that $c(b) \leq c(a)$ and $\Delta_G -a+b \in \cI_{\leq q(\eps)}$. 
\end{claim}
\begin{claimproof}
Let $(C_{r}(\alpha),\cI') =  \left[ (E, \cI) \cap C_{r}(\alpha) \right]_{\leq q(\eps)}$ and $\cm = (C_{r}(\alpha),\cI')$. By Lemma~\ref{lem:minBasis}  it holds that $R \cap C_r(\alpha)$ is a minimum basis w.r.t. $c$ of $\cm$.  Define $D = \{e \in C_r(\alpha) \cap R~|~c(e) \leq c(a)\}$. Then, since $a \in C_r(\alpha) \setminus R$, it holds that $D+a \notin \cI'$ by Lemma~\ref{lem:notIS}. In addition, by Definition~\ref{def:matroids} it holds that $\cI' = \{A \subseteq C_r(\alpha)~|~A \in \cI_{\leq q(\eps)}\}$. Hence, $D+a \notin \cI_{\leq q(\eps)}$. Therefore, by Lemma~\ref{lem:gen} there is $b \in D \setminus \Delta_G$ such that $\Delta_G -a+b \in \cI_{\leq q(\eps)}$. The claim follows since $c(b) \leq c(a)$ because $b \in D$.

\end{claimproof} \newline Using Claim~\ref{claim:auxRep}, let $b \in C_r(\alpha) \cap R \setminus \Delta_G$  such that $c(b) \leq c(a)$ and $\Delta_G -a+b \in \cI_{\leq q(\eps)}$. Then,  the properties of Definition~\ref{def:sub} are satisfied for $Z_G-a+b$ by the following. \begin{enumerate}
\item $(G\setminus H) \cup (Z_G -a +b) =  \Delta_G-a+b \in \cI_{\leq q(\eps)}$ by the definition of $b$. 

\item  $c(Z_G-a+b) \leq c(Z_G) \leq c(G\cap H)$ because $c(b) \leq c(a)$. 

\item for all $r' \in [\log_{1-\eps} \left(\frac{\eps}{2}\right)+1]$ it holds that $|C_{r'}(\alpha) \cap (Z_G-a+b)| = |C_{r'}(\alpha) \cap Z_G| = |C_{r'}(\alpha) \cap G \cap H|$ because $a,b \in C_r(\alpha)$.

\item $|(G \setminus H) \cap (Z_G -a+b)| \leq |(G \setminus H) \cap (Z_G)| = 0$ where the inequality follows because $b \notin \Delta_G$ and the equality is since $Z_G$ is a substitution of $G$. 
\end{enumerate}

 By the above and Definition~\ref{def:sub}, it follows that $Z_G+a-b$ is a substitution of $G$; that is, $Z_G+a-b \in \mathcal{S}(G)$. Moreover, \begin{equation}
\label{eq:ZG}
|R \cap (Z_G -a+b)|>|R \cap Z_G| = \max_{Z \in \mathcal{S}(G)} |Z \cap R|.
\end{equation} The first inequality is because $a \in Z_G \setminus R$ and $b \in R$. By \eqref{eq:ZG} we reach a contradiction since we found a replacement of $G$ with more elements in $R$ than $Z_G \in \mathcal{S}(G)$, which is defined as a replacement of $G$ with maximal number of elements in $R$. Therefore, $Z_G \subseteq R$ as required. 
\end{proof}
We are now ready to prove Lemma~\ref{lem:main}. The proof follows by showing that for any $G \in \cI_{\leq q(\eps)}$ a substitution of $G$ which is a subset of $R$ is in fact a replacement of $G$. 

	\noindent{\bf Proof of Lemma~\ref{lem:main}}: Let $G \in \cI_{\leq q(\eps)}$. By Lemma~\ref{lem:R1}, $G$ has a substitution $Z_G \subseteq R$. Then, \begin{equation}
	\label{eq:profitR}
	\begin{aligned}
		p(Z_G) \geq{} & \sum_{r \in [\log_{1-\eps} \left(\frac{\eps}{2}\right)+1]} p(C_r(\alpha) \cap Z_G) 
		\\ \geq{} &  \sum_{r \in [\log_{1-\eps} \left(\frac{\eps}{2}\right)+1] \text{ s.t. } C_r(\alpha) \neq \emptyset} |C_r(\alpha) \cap Z_G| \cdot \min_{e \in C_r(\alpha)} p(e) 
		\\ \geq{} & \sum_{r \in [\log_{1-\eps} \left(\frac{\eps}{2}\right)+1] \text{ s.t. } C_r(\alpha) \neq \emptyset} |C_r(\alpha) \cap G \cap H| \cdot (1-\eps) \cdot \max_{e \in C_r(\alpha)} p(e) 
		\\ \geq{} & (1-\eps) \cdot p(G \cap H).
	\end{aligned}
\end{equation}  The third inequality is by \eqref{Er} and Property~\ref{pp:3} of Definition~\ref{def:sub}. The last inequality follows from Lemma~\ref{claim:RR}. Therefore,  
\begin{equation}
	\label{eq:profitFINAL}
	\begin{aligned}
	p((G \setminus H) \cup Z_G) & = p(G \setminus H)+p(Z_G) \\
	&  \geq p(G \setminus H)+(1-\eps) \cdot p(G \cap H) \\
	& \geq (1-\eps) \cdot p(G).
	\end{aligned}
\end{equation}
The first equality follows by Property~\ref{pp:4} of Definition~\ref{def:sub}. The first inequality is by \eqref{eq:profitR}. Now,  $Z_G$ satisfies Property~\ref{p:I} and~\ref{p:s} of Definition~\ref{def:Replacement} by Properties~\ref{pp:1},~\ref{pp:2} of Definition~\ref{def:sub}, respectively. In addition, $Z_G$ satisfies Property~\ref{p:p} of Definition~\ref{def:Replacement} by \eqref{eq:profitFINAL}. Finally, $Z_G$ satisfies Property~\ref{p:car}  of Definition~\ref{def:Replacement} by $$|Z_G| = \sum_{r \in [\log_{1-\eps} \left(\frac{\eps}{2}\right)+1]} |Z_G \cap C_r(\alpha)| \leq \sum_{r \in [\log_{1-\eps} \left(\frac{\eps}{2}\right)+1]} |G \cap H \cap C_r(\alpha)| = |G \cap H|.$$ The first inequality holds since $Z_G$ is a substitution of $G$. The last equality follows from Lemma~\ref{claim:RR}. We conclude that $Z_G$ is a replacement of $G$ such that $Z_G \subseteq R$; thus, $R$ is a representative set by Definition~\ref{def:Representatives}. By Step~\ref{step:Br} of Algorithm~\ref{alg:representative} and Definition~\ref{def:matroids}, it holds that $|R \cap C_r(\alpha)| \leq q(\eps)$; thus, it follows that $|R| \leq \left( \log_{1-\eps} \left(\frac{\eps}{2}\right)+1\right) \cdot q(\eps)$. %

 We now bound the time complexity of Algorithm \textsf{FindRep}. Step~\ref{step:APP} can be computed in time $\textnormal{poly}(|\ck|)$ using a PTAS for BMI with
 an error parameter
 $\eps= \frac{1}{2}$ (see, e.g.,~\cite{GZ10}). In addition, a membership oracle for the matroid in Step~\ref{step:Br} can be implemented in time $\textnormal{poly}(|\ck|)$ by Definition~\ref{def:matroids}. Finally, the basis in Step~\ref{step:Br} can be computed in time $\textnormal{poly}(|\ck|)$ using a greedy matroid basis minimization algorithm \cite{cormen2022introduction}. Hence, the running time of the algorithm is $\textnormal{poly}(|\ck|)$. \qed%

			\section{Discussion}
\label{sec:discussion}
In this paper we showed that the budgeted matroid independent set problem admits an EPTAS, thus improving upon the previously known  schemes for the problem. The speed-up is 
achieved by replacing the exhaustive enumeration used by existing algorithms \cite{GZ10,BBGS11,CVZ11,FKS21} with efficient enumeration
over subsets of a {\em representative set} whose size depends  only on $1/\eps$.

The representative set found by our algorithm is a minimum cost matroid basis for a union matroid. The union matroid is a union of a  matroid for each  {\em profit class} in the given instance. 
The basis itself can be easily found using a simple 
greedy procedure. The correctness relies on matroid exchange properties, optimality properties of minimal cost bases and a ``profit gap'' 
obtained by discarding a subset of elements from an optimal solution.

We note that our EPTAS, which directly exploits structural properties of our problem, already achieves substantial improvement over schemes developed for generalizations of BMI.
Furthemore, we almost resolve the complexity status of the problem  w.r.t approximation schemes. The existence of an FPTAS remains open.  

The notion 
of representative sets can potentially be used to replace exhaustive enumeration in the  PTASs for  {\em multiple knapsack with matroid constraint}  \cite{FKS21} 
and {\em budgeted matroid intersection} \cite{BBGS11}, 
leading to EPTASs for both problems.
It seems that the construction of a representative set, as well as the ideas behind the main lemmas, can be adapted to the setting of the  multiple knapsack with matroid problem. However, to derive an EPTAS for budgted matroid intersection, one needs to generalize first the 
concept of representative set, so it can be applied to matroid intersection. We leave this generalization as an interesting direction for follow-up work.

		\bibliographystyle{splncs04}
	
		\bibliography{bibfile}
		
\appendix
\section{Solving the Linear Program}
\label{sec:LP}

In this section we show how $\textnormal{LP}(\ck,F,\alpha)$ can be solved in polynomial time, thus proving Lemma~\ref{lem:LPrunning}.

The main idea is to show that the feasibility domain of  $\textnormal{LP}(\ck,F,\alpha)$  is a  {\em well described} polytope, and that a {\em separation oracle} for it can be implemented in polynomial time.  Thus, using a known connection between separation and optimization, we obtain a polynomial-time algorithm which solves $\textnormal{LP}(\ck,F,\alpha)$. Our initial goal, however, is to show that $\textnormal{LP}(\ck,F,\alpha)$  is indeed a linear program, for which we need a characterization of the matroid polytope via linear inequalities.

Let $\cN=(\Omega, \cJ)$ be a matroid. The {\em matroid rank function of $\cN$}, $\rank_{\cN}:2^{\Omega} \to \mathbb{N}$, is defined by $\rank_{\cN}(S) = \max \left\{ |T|~ \middle|~T\in \cJ,~T\subseteq S\right\}$. 
That is, $\rank_{\cN}(S)$ is the maximal size of an independent set which is also a subset of $S$. The rank function is used in a characterization of a matroid polytope $P_{\cN}$ via a system of linear inequalities.
\begin{theorem}[Corollary 40.2b in \cite{schrijver2003combinatorial}] 
	\label{thm:poly_char}
	For any matroid $\cN=(\Omega, \cJ)$, it holds that
	$$P_{\cN} = \left\{ \bx \in \mathbb{R}_{\geq 0 }~\middle|~\forall S\subseteq \Omega:~~\sum_{\omega\in S} \bx_{\omega }\leq \rank_{\cN} (S)  \right\}.$$
\end{theorem} 

Let $\ck=(E,\II, c,p,\beta)$ be a BMI instance, $F\subseteq E$ and $\frac{\OPT(\ck)}{2} \leq \alpha \leq \OPT(\ck)$. By Theorem~\ref{thm:poly_char} it holds that \eqref{LP} is equivalent to the following linear program.
	\begin{alignat}{4}
		\textnormal{LP}(\ck,F,\alpha)~~~~~~~ \max\quad       ~~ & \sum_{e \in E(\alpha)\setminus F} \bar{x}_{e}   \cdot p(e)  &&        \nonumber                                                \\
		\textsf{s.t.\quad} & \sum_{e \in E(\alpha)\setminus F} \bar{x}_{e}  \cdot c(e)         \leq  \beta-c(F) ~~~~~	\label{LP:budget}           \\
		&\sum_{e\in S} \bx_e \leq \rank_{\cm_F(\alpha)} (S) &&\forall S\subseteq E(\alpha)\setminus F \label{LP:rank} \\
		&\bar{x}_e \geq 0&&\forall e \in E(\alpha) \setminus F\label{LP:nonneg}
\end{alignat}  

We follow the definitions and techniques presented in \cite{GLS93}.
We say a polytope  $P\subseteq \mathbb{R}^n$ is of {\em facet complexity $\varphi$} if it is the solution set of a system of linear inequalities with rational coefficients, and the encoding length of each inequality in the system is at most $\varphi$. The following is an immediate consequence of the representation of  $\textnormal{LP}(\ck,F,\alpha)$ as a linear program.
\begin{observation}
	\label{obs:facet_comp}
The feasibility region of 
$\textnormal{LP}(\ck,F,\alpha) $ is of facet complexity polynomial in $|\ck|$.  
\end{observation}

A separating hyperplane between a polytope $P\subseteq \mathbb{R}^n$ and $\bx \in \mathbb{R}^n$ is  vector $\bc\in \mathbb{R}^n$ such that $\bc \cdot \bx > \max_{\by\in P } \bc\cdot \by$, where $\bc \cdot \by$ is the inner product  of $\bc$ and $\by$. With a slight abuse of notation, we say that the constraint $\bc \cdot \bz \leq L$, where $\bc \in \mathbb{R}^n$ and $L\in \mathbb{R}$, is a separation hyperplane between $P$ and $\bx$ if $\bc \cdot \bx >L$ and $\bc \cdot \by\leq L$ for every $\by \in P$. 
A {\em separation oracle} for a polytope $P\subseteq \mathbb{R}^n$ is an oracle which receives $\bx \in \mathbb{R}^n$ as input, and either determines that $\bx\in P$ or returns a separating hyperplane between $P$ and $\bx$.  We use the following known connection between separation and optimization.
\begin{theorem}[Thm 6.4.9 and Remark	6.5.2 in \cite{GLS93}]
	\label{thm:ellipsoid} There is an algorithm which given $n,\varphi\in \mathbb{N}$,  $\bc \in \mathbb{R}^n$ and a separation oracle for a non-empty polytope $P\subseteq \mathbb{R}^n$ of facet complexity $\varphi$, returns a vertex $\by$ of $P$ such that $\bc\cdot \by =\max_{\bx \in P} \bc\cdot \bx$ in time polynomial in $n$ and $\varphi$. 
\end{theorem}

Thus, to show 	$\textnormal{LP}(\ck,F,\alpha)$ can be solved in polynomial time, we need to show that a separation oracle for $\textnormal{LP}(\ck,F,\alpha)$  can be implemented in polymomial time. To this end, we use a known result for matroid polytopes.
\begin{theorem}[Thm 40.4 in \cite{schrijver2003combinatorial}]
	\label{thm:matroid_separation}
There is a polynomial time algorithm $\textsf{MatroidSeparator}$ which given a subset of elements $\Omega$, a membership oracle for a matroid $\cN=(\Omega,\cJ)$ and $\bx \in \mathbb{Q}_{\geq 0}^\Omega$ either determines that $\bx \in P_{\cN}$ or returns  $S\subseteq \Omega$ such that $\sum_{\omega \in S} \bx_{\omega} >\rank_{\cN}(S)$. 
	\end{theorem}
A polynomial-time implementation of a separation oracle for 	$\textnormal{LP}(\ck,F,\alpha)$ is now straightforward.
\begin{lemma}
	\label{lem:separation}
There is an algorithm which given a BMI instance  $\ck=(E,\II, c,p,\beta)$,  $F\subseteq E$, $\frac{\OPT(\ck)}{2} \leq \alpha \leq \OPT(\ck)$ and $\bx\in \mathbb{R}^{E}$ implements a separation oracle for the feasibility region of 	$\textnormal{LP}(\ck,F,\alpha)$  in polynomial time.
\end{lemma}
\begin{proof}
To implement a separation oracle, the algorithm first checks if the input $\bx$ satisfies constraints~\eqref{LP:budget} and~\eqref{LP:nonneg}. If one of the constraints is violated, the algorithm returns the constraint as a separating hyperplane.

Next, the algorithm invokes $\textsf{MatroidSeparator}$  (Theorem~\ref{thm:matroid_separation}) with the matroid $\cm_{F} (\alpha)$ and the point $\bx$. If  $\textsf{MatroidSeparator}$ returns that $\bx\in P_{\cm_{F}(\alpha)}$ then the algorithm returns that $\bx$ is in the feasibility region. Otherwise, the $\textsf{MatroidSeparator}$ returns $S\subseteq E(\alpha)\setminus F$ such that $\sum_{e\in S} \bx_e > \rank_{\cm_F(\alpha)} (S) $; that is, a set $S$ for which \eqref{LP:rank} is violated. In this case, the algorithm returns $\mathbbm{1}^S$ as a separating hyperplane. 
Observe that for every $\by$ in the feasibility region of $\textnormal{LP}(\ck,F,\alpha)$, it holds that 
$$\mathbbm{1}^{S} \cdot \by = \sum_{e\in S} \by_e  \leq \rank_{\cm_F(\alpha)} (S)  < \mathbbm{1}^{S} \cdot \bx,$$
where the first inequality is by \eqref{LP:rank}. Thus, $\mathbbm{1}^S$ is indeed a separating hyperplane.  

Clearly, the separating hyperplanes returned by the algorithm are indeed separating hyperplanes. Furthermore, if the algorithm asserts that $\bx$ is in the feasibility region then constraints \eqref{LP:budget} and \eqref{LP:nonneg} hold as those were explicitly checked, and \eqref{LP:rank} holds by $\bx\in P_{\cm_F(\alpha)}$ (Theorem~\ref{thm:matroid_separation} and Theorem~\ref{thm:poly_char}). That is, $\bx$ is indeed in the feasibility region of $\textnormal{LP}(\ck,F,\alpha)$. The algorithm runs in polynomial-time as each of its steps can be implemented in polynomial-time.
\end{proof}

\begin{proof}[{\bf Proof of Lemma~\ref{lem:LPrunning}}] Observe that the vector $(0,\ldots,0)\in \mathbb{R}^{E(\alpha)}$ is in the feasibility region of $\textnormal{LP}(\ck,F,\alpha)$, and thus the feasibility region is not empty. By Theorem~\ref{thm:ellipsoid} and Observation~\ref{obs:facet_comp}, there is an  algorithm which finds an optimal  basic solution for  $\textnormal{LP}(\ck,F,\alpha)$ using a polynomial number of operations and calls for a separation oracle. By Lemma~\ref{lem:separation}, the separation oracle can be implemented in polynomial time as well. Thus, an optimal basic solution for $\textnormal{LP}(\ck,F,\alpha)$ can be found in time polynomial in $|\ck|$.  
	\end{proof}
	\end{document}